\pdfoutput=1

\documentclass[conference]{IEEEtran}

\usepackage[hidelinks=true]{hyperref}  
\usepackage{caption}
\usepackage{subcaption}
\usepackage{graphicx}
\usepackage{xcolor}
\usepackage{booktabs, makecell}
\usepackage{paralist}
\usepackage{fancybox}
\usepackage{background}

\usepackage{float}
\usepackage{amsmath, amssymb, amsthm}
\usepackage{mathabx}
\usepackage{wasysym}
\usepackage{mathtools}
\usepackage{algorithmicx}
\usepackage{algorithm}
\usepackage{algpseudocode}
\usepackage{commath}
\usepackage{tikz}
\usetikzlibrary{automata,positioning,arrows}
\usetikzlibrary{decorations.markings}
\usetikzlibrary{calc,fit}
\usepackage[linguistics]{forest}
\forestset{
	colour my roots/.style={
		before typesetting nodes={
			edge=#1,
			for ancestors={
				edge=#1,
				#1,
			},
			#1,
		}
	},
	my edge label/.style={
		edge label={
			node [midway, fill=white, font=\scriptsize] {#1}
		}
	}
}
\usepackage{enumitem}
\usepackage{etex}
\newlist{longenum}{enumerate}{5}
\setlist[longenum,1]{label=\roman*)}
\setlist[longenum,2]{label=\alph*)}
\setlist[longenum,3]{label=\arabic*)}
\setlist[longenum,4]{label=(\roman*)}
\setlist[longenum,5]{label=(\alph*)}

\makeatletter
\tikzset{%
	remember picture with id/.style={%
		remember picture,
		overlay,
		save picture id=#1,
	},
	save picture id/.code={%
		\edef\pgf@temp{#1}%
		\immediate\write\pgfutil@auxout{%
			\noexpand\savepointas{\pgf@temp}{\pgfpictureid}}%
	},
	if picture id/.code args={#1#2#3}{%
		\@ifundefined{save@pt@#1}{%
			\pgfkeysalso{#3}%
		}{
			\pgfkeysalso{#2}%
		}
	}
}

\def\savepointas#1#2{%
	\expandafter\gdef\csname save@pt@#1\endcsname{#2}%
}

\def\tmk@labeldef#1,#2\@nil{%
	\def\tmk@label{#1}%
	\def\tmk@def{#2}%
}

\tikzdeclarecoordinatesystem{pic}{%
	\pgfutil@in@,{#1}%
	\ifpgfutil@in@%
	\tmk@labeldef#1\@nil
	\else
	\tmk@labeldef#1,(0pt,0pt)\@nil
	\fi
	\@ifundefined{save@pt@\tmk@label}{%
		\tikz@scan@one@point\pgfutil@firstofone\tmk@def
	}{%
		\pgfsys@getposition{\csname save@pt@\tmk@label\endcsname}\save@orig@pic%
		\pgfsys@getposition{\pgfpictureid}\save@this@pic%
		\pgf@process{\pgfpointorigin\save@this@pic}%
		\pgf@xa=\pgf@x
		\pgf@ya=\pgf@y
		\pgf@process{\pgfpointorigin\save@orig@pic}%
		\advance\pgf@x by -\pgf@xa
		\advance\pgf@y by -\pgf@ya
	}%
}
\newcommand\tikzmark[2][]{%
	\tikz[remember picture with id=#2] #1;}
\makeatother

\newcommand\MyBox[4][-1ex]{%
	\tikz[remember picture,overlay,pin distance=0cm]
	{\draw[draw=white,line width=1pt,fill=#4!15,rectangle,rounded corners]
		( $ (pic cs:#2) + (-1ex,2ex) $ ) rectangle ( $ (pic cs:#3) + (1ex,#1) $ );
	}
}

\newcommand*\redlaptop{\includegraphics[height=\heightof{M}]{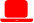}}

\makeatletter
\def\th@remark{%
	\thm@headfont{\bfseries}%
	\normalfont 
	\thm@preskip\topsep \divide\thm@preskip\tw@
	\thm@postskip\thm@preskip
}
\g@addto@macro\th@remark{\thm@headpunct{}}
\makeatother
\theoremstyle{remark}

\newtheorem{defn}{Definition} 
\newtheorem{theorem}{Theorem}
\newtheorem{lemma}{Lemma}

\newtheorem{innercustomgeneric}{\customgenericname}
\providecommand{\customgenericname}{}
\newcommand{\newcustomtheorem}[2]{%
	\newenvironment{#1}[1]
	{%
		\renewcommand\customgenericname{#2}%
		\renewcommand\theinnercustomgeneric{##1}%
		\innercustomgeneric
	}
	{\endinnercustomgeneric}
}
\newcustomtheorem{customthm}{Theorem}
\newcustomtheorem{customlemma}{Lemma}

\newcommand\stutt{\mathrel{\overset{\makebox[0pt]{\mbox{\normalfont\tiny\sffamily st}}}{\equiv}}}

\newcommand{\ignore}[1]{}

\makeatletter
\renewcommand*{\ALG@name}{Controller Program CP\!}
\makeatother

\algrenewcommand\textproc{}
\algrenewcomment[1]{\(\!\!\!\textcolor{gray}{//}\) #1}
\def\thesubsectiondis{{\thesection}-\Alph{subsection}}

\begin{document}
\SetBgContents{This is an extended version of the same work (with the same name), which appears in FMCAD 2020}
\SetBgPosition{current page.center}
\SetBgAngle{0}                                    
\SetBgScale{1.2}                                  
\SetBgHshift{0}                                   
\SetBgVshift{11.3cm} 

\title{Model Checking Software-Defined Networks with Flow Entries that Time Out}

\author{\IEEEauthorblockN{Vasileios Klimis, George Parisis and Bernhard Reus}
\IEEEauthorblockA{University of Sussex, UK\\
\{v.klimis, g.parisis, bernhard\}@sussex.ac.uk}
}

\maketitle

\begin{abstract}
Software-defined networking (SDN) enables advanced operation and management of network deployments through (virtually) centralised, programmable controllers, which deploy network functionality by installing rules in the flow tables of network switches. Although this is a powerful abstraction, buggy controller functionality could lead to severe service disruption and security loopholes, motivating the need for \mbox{(semi-)}automated tools to find, or even verify absence of, bugs. Model checking SDNs has been proposed in the literature, but none of the existing approaches can support dynamic network deployments, where flow entries expire due to timeouts. This is necessary for automatically refreshing (and eliminating stale) state in the network (termed as \emph{soft-state} in the network protocol design nomenclature), which is important for scaling up applications or recovering from failures. In this paper, we extend our model (MoCS) to deal with timeouts of flow table entries, thus supporting soft state in the network. Optimisations are proposed that are tailored to this extension. We evaluate the performance of the proposed model in {\scshape Uppaal} using a load balancer and firewall in network topologies of varying size.
\end{abstract}

\IEEEpeerreviewmaketitle

\section{Introduction}
\label{intro}

Software-defined networking (SDN) \cite{TheRoadToSDN} revolutionised network operation and management along with future protocol design; a virtually centralised and programmable controller `programs' network switches through interactions (standardised in OpenFlow \cite{Openflow}) that alter switches' flow tables. In turn, switches push packets to the controller when they do not store state relevant to forwarding these packets. Such a paradigm departure from traditional networks enables the rapid development of advanced and diverse network functionality; e.g., in designing next-generation inter-data centre traffic engineering \cite{DevoFlow}, load balancing \cite{Plug-n-Serve}, firewalls \cite{SDNFirewall} and Internet exchange points (IXPs) \cite{SDX}. Although this is a powerful abstraction, buggy controller functionality could lead to severe service disruption and security loopholes. This has led to a significant amount of research on SDN verification and/or bug finding, including static network analysis \cite{Anteater,HSA,VeriCon2}, dynamic real-time bug finding \cite{McClurg,Plotkin,Deltanet,Netplumber}, and formal verification approaches, including symbolic execution \cite{SymNet,NICE,NetSMC} and model checking \cite{mocs,McClurg, NetSMC,Kuai} methods. A comprehensive review of existing approaches along with their shortcomings can be found in \cite{survey_NV}.

Model checking is a renowned automated technique for hardware and software verification and existing model checking approaches for SDNs have shown promising results with respect to scalability and model expressivity, in terms of supporting realistic network deployments and the OpenFlow standard. However, a key limitation of all existing approaches is that they cannot model forwarding state (added in network switches’ flow tables by the controller) that expires and gets deleted. Without this, one cannot model nor verify the correctness of SDNs with soft-state which is prominent in the design of protocols and systems that are resilient to failures and scalable; e.g., as in \cite{hedera}, where flow scheduling is on a per-flow basis, and numerous network protocols where in-network state is not explicitly removed but expires, so that overhead is minimised \cite{HardSoft}.

In this paper, we extend our model (MoCS) \cite{mocs} to support soft-state, complying with the OpenFlow specification, by allowing flow entries to time out and be deleted. We propose relevant optimisations (as in \cite{mocs}) in order to improve verification performance and scalability. We evaluate the performance of the proposed model extensions in {\scshape Uppaal} using a load balancer and firewall in network topologies of varying size.

\section{MoCS SDN Model}
\label{model}

The MoCS model \cite{mocs} is formally defined by means of an action-deterministic transition system. We parameterise the model by the underlying network topology, $\lambda$, and the controller program, {\sc cp}, in use. The model is a 6-tuple $\mathcal{M}_{(\lambda,\textsc{cp})} = (S, s_0, A, \hookrightarrow, \mathit{AP}, L)$, where $S$ is the set of all states the SDN may enter, $s_0$ the initial state, $A$ the set of actions which encode the events the network may engage in, $\hookrightarrow \subseteq S \times A \times S$ the transition relation describing which execution steps the system undergoes as it perform actions, \emph{AP} a set of atomic propositions describing relevant state properties, and $L: S \to 2^{\mathit{AP}}$ is a labelling function, which relates to any state $s\in S$ a set $L(s) \in 2^{\mathit{AP}}$ of those atomic propositions that are true for $s$. Such an SDN model is composed of several smaller systems, which model network components (hosts\footnote{A host can act as a client and/or server.}, switches and the controller) that communicate via queues and, combined, give rise to the definition of $\hookrightarrow$. 
A detailed description of MoCS' components and transitions can be found in \cite{mocs}. Due to lack of space, in this paper, we only discuss aspects of the model that are required to understand and verify the soundness of the proposed model extensions, and examples used in the evaluation section. Figure \ref{fig:model} illustrates a high-level view of OpenFlow interactions, modelled actions and queues, including the proposed extensions discussed in Section~\ref{sec:extensions}. 

\begin{figure}[t]
	\centering
	\includegraphics[width=1\linewidth]{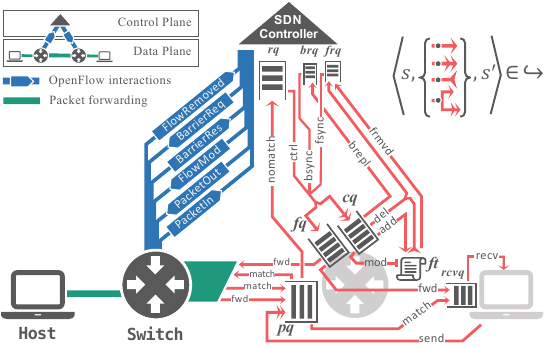}
	\caption{A high-level view of OpenFlow interactions (left half) and modelled actions (right half). A red solid-line arrow depicts an action which, when fired, (1) dequeues an item from the queue the arrow begins at, and (2) (possibly) adds an item in the queue the arrowhead points to (or multiple items if the arrow is double-headed). Deleting an item from the target queue is denoted by a reverse arrowhead; modifying in, by a hammerhead. A forked arrow denotes (possibly) multiple targeted queues.}
	\label{fig:model}
		\vspace{-.3cm}
\end{figure}

%

\noindent\textbf{States and queues:} A \emph{state} is a triple $(\pi,\delta, \gamma)$, where $\pi$ is a family of hosts, each consisting of a receive queue (\emph{rcvq}); $\delta$ is a family of switches, consisting of a switch packet queue (\emph{pq}), switch forward queue (\emph{fq}), switch control queue (\emph{cq}), switch flow table (\emph{ft}); $\gamma$ consists of the local controller program state $\mathit{cs \in CS}$, and a family of controller queues: request queue (\emph{rq}), barrier-reply queue (\emph{brq}) and flow-removed queue (\emph{frq}). So $\pi$ and $\delta$ describe   the \emph{data-plane}, and $\gamma$ the \emph{control plane}. The network components communicate via the shared queues. Each transition models a certain network event that will involve some of the queues, and maybe some other network state. Concurrency is modelled through interleavings of those events.

\noindent\textbf{Transitions:} Each transition is labelled with an action $\alpha \in A$ that indicates the nature of the network event. We write $s \xhookrightarrow[]{  \alpha  } s'$ and $(s, \alpha, s') \in \hookrightarrow$ interchangeably to denote that the network moved from state $s$ to $s'$ by executing transition $\alpha$. The parts of the network involved in each individual $\alpha$, i.e.\ \emph{packets, rules, barriers, switches, hosts, ports} and \emph{controller states}, are included in the transition label as parameters; e.g., $\mathit{match}(sw,pkt,r)\in A$ denotes the action that switch $sw$ matches packet $pkt$ by rule $r$ and, as a result, forwards it accordingly, leading to a new state after transition.

%
\noindent\textbf{Atomic propositions:} The propositions in \emph{AP} are statements on 
\begin{inparaenum}[(1)]
	\item controller program states, denoted by $Q(q)$ which expresses that the controller program is in state $\mathit{q \in CS}$, allowing one to reason about the controller's internal data structures, and  
	\item packet header fields -- those packets may be in any switch buffer  \emph{pq} or host buffer \emph{rcvq} (but no other buffers). For instance, $\exists \mathit{pkt} {\in} \mathit{sw.pq}\,.\, P(\mathit{pkt})$ is a legitimate atomic proposition that states that there is a packet in $\mathit{sw}$'s packet queue that satisfies packet \emph{pkt} property $P$. 
\end{inparaenum}

\noindent\textbf{Topology: }$\lambda$ describes the network topology as a bijective map which associates one network interface (a pair of networking device and physical port) to another.

\noindent\textbf{Specification Logic:} The properties of the SDNs to be checked in this paper are safety properties, expressed in linear-time temporal logic without `next-step' operator, $\text{LTL}_{ \setminus \{\bigcirc \}}$.
We have enriched the logic by modal operators of \emph{dynamic logic} \cite{DL1}, allowing formula construct of the form $[\alpha(\vec{x})]P$ stating that whenever an event $\alpha(\vec{x})$ happened, $P$ must hold. Note that $P$ may contain variables from $x$. This extension is syntax sugar in the sense that the formulae may be expressed by additional state
; e.g., $\big[ \mathit{match(sw,pkt,r)}  \big](r.\mathit{fwdPort}=\texttt{\small drop})$ states that if \emph{match} happened, it was via a rule that dropped the packet.
This permits specification formulae to be interpreted not only over states, but also over actions that have happened.  

%

\noindent The model checking problem then, for an SDN model $\mathcal{M}_{(\lambda,\textsc{cp})}$ with a given topology $\lambda$, a control program $\textsc{cp}$ and a formula $\varphi$ of the specification logic as described above, boils down to checking whether all runs of $\mathcal{M}_{(\lambda,\textsc{cp})}$ satisfy  $\varphi$, short  $\mathcal{M}_{(\lambda,\textsc{cp})} \models \Box\varphi$.

%

\noindent\textbf{SDN Operation:} End-hosts send and receive packets (\textit{send} and \textit{recv} actions in Figure \ref{fig:model}) and switches process incoming packets by matching them (or failing to) with a flow table entry (rule). In the former case (\textit{match} action), the packet is forwarded as prescribed by the rule. In the opposite case (\textit{nomatch} action), the packet is sent to the controller (\textit{PacketIn} message on the left side of Figure \ref{fig:model}). The controller’s packet handler is executed in response to incoming \textit{PacketIn} messages; as a result of its execution, its local state may change, a number of packets (\textit{PacketOut} message) and rule updates (\textit{FlowMod} message), interleaved with barriers (\textit{BarrierReq} message), may be sent to network switches. Network switches react to incoming controller messages; they forward packets sent by the controller as specified in the respective \textit{PacketOut} message (\textit{fwd} action), update their own forwarding tables (\textit{add/del} actions), respecting set barriers and notifying the controller (\textit{BarrierRes} message) when said barriers are executed (\textit{brepl} action). Finally, upon receiving a \textit{BarrierRes} message, the controller executes the respective handler (\textit{bsync} action), which can result in the same effects as the \textit{PacketIn} message handler. 

\noindent\textbf{Abstractions:} To obtain finitely representable states, all queues in the model must be finitely representable. For packet queues we use multisets, subject to $(0,\infty)$ abstraction \cite{01infty}; a packet either does not appear in the queue or appears an unbounded number of times. The other queues are simply modelled as finite sets. Modelling queues as sets means that entries are not processed in the order of arrival. This is intentional for packet queues but for controller queues this may limit behaviour unless the controller program is order-insensitive. We focus on those controller programs in this paper.

\section{Modelling Flow Entry Timeouts}
\label{sec:extensions}

In order to model soft-state in the network, we enrich our model with two new actions that model flow entry timeouts and subsequent handling of these timeouts by the controller program. Note that in our model, timeouts are not triggered by any kind of clock; instead, they are modelled through the interleaving of actions in the underlying transition system that ensure that flow removal (and subsequent handling by the controller program) will appear as it would for any possible value of a timeout in a real system. 

The new actions are defined as follows: $\mathit{frmvd(sw,r)}$ models the timeout event, as an action in the transition system that removes the flow entry (rule) $r$ from switch $\mathit{sw}$ and notifies the controller by placing a \textit{FlowRemoved} message (see Figure \ref{fig:model}) in the respective queue (\textit{frq}). The $\mathit{fsync(sw,r,cs)}$ action models the call to the \textit{FlowRemoved} message handler. As a result of the handler execution, the controller's local state (\textit{cs}) may change, a number of packets (\textit{PacketOut} messages) and rule updates (\textit{FlowMod} messages), interleaved with barriers (\textit{BarrierReq} message), may be sent to network switches. 
In order to model timeouts, rules are augmented with a \emph{timeout} bit which, when true, signals that the installed rule can be removed at any time, i.e., the \emph{frmvd}-action can be interleaved, in any order, with any other action that is enabled at any state later than the installation of this rule. 

To support our examples, we add to the set of \textit{FlowMod} messages a \emph{modify flow entry} instruction. In \cite{mocs} we only used $\mathit{add(sw,r)}$ and $\mathit{del(sw,r)}$ messages, for installing and deleting rule $r$ at switch $\mathit{sw}$, respectively.
We now add $\mathit{mod(sw,f,a)}$ to these messages. This instructs switch $\mathit{sw}$ that if a rule is found in $\mathit{sw.ft}$ that matches field \emph{f}, its forwarding actions are modified by \emph{a}. If no such rule exists, $\mathit{mod}(\cdot)$ does not do anything.

\noindent\textbf{Optimisation:}
To tackle the state-space explosion, we exploit the fact that some traces are observationally (w.r.t.\ the property to be proved) equivalent, so that only one of those needs to be checked. This technique, referred to as \emph{partial-order reduction} (POR) \cite{representatives}, reduces the number of interleavings (traces) one has to check. To prove equivalence of traces, one needs actions to be permutable and invisible to the property at hand. This is the motivation for the following definition:

\begin{defn}[{\sc Safe Actions}] \label{def:safe} Given a context {\sc ctx} $= (\textsc{cp},\lambda,\varphi)$, and SDN model $\mathcal{M}_{(\lambda,\textsc{cp})} = (S, A, \hookrightarrow, s_0, \mathit{AP}, L)$, an action $\alpha(\cdot)\in A(s)$ is called \emph{safe} if it is 
	\begin{inparaenum}[(1)]
		\item \emph{independent} of any other action $\beta$ in $A$, i.e. executing $\alpha$ after $\beta$ leads to the same state as running $\beta$ after $\alpha$, and 
		\item \emph{unobservable} for $\varphi$ (also called $\varphi$-\emph{invariant}), i.e., $s\models \varphi$ iff $\alpha(s)\models \varphi$ for all $s\in S$ with $\alpha\in A(s)$.
	\end{inparaenum}
\end{defn}
 The following property of controller programs is needed to show safety:
 
\begin{defn}[{\sc Order-sensitive Controller Program}]
	\label{def:ord-sens}
	A controller program {\sc cp} is order-sensitive if there exists a state $s\in S$ and two actions $\alpha,\beta$ in $\{ \mathit{ctrl}(\cdot), \mathit{bsync}(\cdot), \mathit{fsync}(\cdot)   \}$ such that $\alpha,\beta \in A(s)$
	and $s\xhookrightarrow[]{\alpha} s_1 \xhookrightarrow[]{\beta} s_2$ and $s\xhookrightarrow[]{\beta} s_3 \xhookrightarrow[]{\alpha} s_4$ with $s_2\neq s_4$.
\end{defn}

In \cite{mocs} we already showed that certain actions are safe and can be used for PORs. We now show that the new $\mathit{fsync}(\cdot)$ action is safe on certain conditions.

\begin{lemma}[{\sc Safeness Predicates for} \emph{fsync}]\label{lemma:safe} For transition system $\mathcal{M}_{(\lambda,\textsc{cp})} = (S, A, \hookrightarrow, s_0, \mathit{AP}, L)$ and a formula $\varphi\in \text{LTL}_{ \setminus \{\bigcirc \}}$, 
	$\alpha = \mathit{fsync(sw, r, cs)}$ is safe iff the following two conditions are satisfied:

\begin{quote}
\begin{description}[leftmargin=!,labelwidth=\widthof{\bfseries Independence}]
	\item[Independence] {\sc cp} is not order-sensitive 
	\item[Invisibility] {\upshape\texttt{if} }$Q(q)$ in $\mathit{AP}$ occurs in $\varphi$, {\upshape\texttt{then}} $\alpha$ is $\varphi$-invariant
\end{description}	
\end{quote}	
\end{lemma}

\begin{proof} 
	See Appendix \ref{appendix:proofs}.
\end{proof}

Given a context {\sc ctx} $= (\textsc{cp},\lambda,\varphi)$ and an SDN network model $\mathcal{M}_{(\lambda,\textsc{cp})} = (S, A, \hookrightarrow, s_0, \mathit{AP}, L)$, for each state $s \in S$ define $\mathit{ample}(s)$ as follows: $\text{if } \{\alpha\in A(s)\ |\ \alpha \text{ safe } \} \neq \emptyset$, then $\mathit{ample}(s)= \{\alpha\in A(s)\ |\ \alpha \text{ safe } \}$; otherwise $\mathit{ample}(s)= A(s)$.
Next, we define $\mathcal{M}_{(\lambda,\textsc{cp})}^{\mathit{fr}} = (S^{\mathit{fr}}, A, \hookrightarrow_{\mathit{fr}}, s_0, \mathit{AP}, L^{\mathit{fr}})$, where
$S^{\mathit{fr}} \subseteq S$ the set of states reachable from the initial state $s_0$ under $\hookrightarrow_{\mathit{fr}}$, $L^{\mathit{fr}}(s) = L(s)$ for all $s \in S^{\mathit{fr}}$ and  $\hookrightarrow_{\mathit{fr}}\,    \subseteq S^{\mathit{fr}}\times A\times S^{\mathit{fr}}$ is defined inductively by the rule:
\[
\frac {s~ \xhookrightarrow[]{ \alpha}    ~s'}{ s~ \xhookrightarrow[]{ \alpha}_{\mathit{fr}}    s'  }
\qquad   \textup{if }\  \alpha \in \mathit{ample}(s)  \]

Now we can proceed to extend the POR Theorem of \cite{mocs}:

\begin{theorem}[{\sc Flow-Removed Equivalence}]
	Given a property $\varphi \in \text{LTL}_{ \setminus \{\bigcirc \}}$, it holds that $\mathcal{M}_{(\lambda,\textsc{cp})}^{\mathit{fr}}$ satisfies $\varphi$ iff $\mathcal{M}_{(\lambda,\textsc{cp})}$ satisfies $\varphi$. 
	\label{thm:frmvd}
\end{theorem}
The proof is a consequence of Lemma~\ref{lemma:safe} applied to the proof of Theorem~2 in \cite{mocs}. See Appendix \ref{appendix:proofs} for a detailed proof.

\section{Experimental Evaluation}
\label{eval}

In this section we experimentally evaluate the proposed extensions in terms of verification performance and scalability. We use a realistic controller program that enables a network switch to act both as a load balancer and stateful firewall (see \S\ref{subsection:cp}-CP\ref{alg:fwlb}). The load balancer keeps track of the active sessions between clients and servers in the cluster (see Figure \ref{fig:lb}), while, at the same time, only allowing specific clients to access the cluster. Soft state is employed here so that flow entries for completed sessions (that were previously admitted by the firewall) time out and are deleted by the switch without having to explicitly monitor the sessions and introduce unnecessary signalling (and overhead). In the underlying SDN model, the \emph{frmvd} action is fired, which, in turn, deletes the flow entry from the switch's table and notifies the controller of that. This enables the \emph{fsync} action that calls the flow removal handler. 

\begin{figure}[h] 
	\centering
	\includegraphics[width=.9\linewidth]{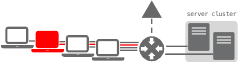}  
	\caption{Four clients and two servers connecting to an OF-switch. $\protect\redlaptop$ is not white-listed.}
	\label{fig:lb}
\end{figure}

A session is initiated by a client which sends a packet (\emph{pkt} in \S\ref{subsection:cp}-CP\ref{alg:fwlb}) to a known cluster address; servers are not directly visible to the client. Sessions are bi-directional therefore the controller must install respective rules to the switch to allow traffic to and from the cluster. The property that is checked here is that 
	\begin{inparaenum}[(1)]
		\item the traffic (i.e. number of sessions, assuming they all produce similar traffic patterns), and resulting load, is uniformly distributed to all available servers, and 
		\item that traffic from non-whitelisted clients is blocked. 
		\end{inparaenum}
		More concretely, ``\textit{a packet from a `dodgy' address should never reach the servers, and the difference between the number of assigned sessions at each server should never be greater than 1}", formally, 
\begin{equation} \tag{$\varphi$} \label{eq:property}
\begin{aligned}[l]
 \Box ~\big(  \forall s_i,s_j \in \mathit{Servers} \!~\forall \mathit{pkt} \in \mathit{s_i.rcvq} \!~&.&\\
  \neg \mathit{pkt.src} = \mathit{dodgy} \!~\land\!~ \abs{\mathit{sLoad}[s_i] - \mathit{sLoad}[s_j]} &<&2 \big)
 \end{aligned}
 \end{equation} where $\mathit{sLoad}$ stores the active session count for each server.

In the first (buggy) version of the controller's packet handler (shaded grey in \S\ref{subsection:cp}-CP\ref{alg:fwlb}) and flow removal handler \S\ref{subsection:cp}-CP\ref{alg:naive}, the controller program assigns new sessions to servers in a round-robin fashion and keeps track of the active sessions (array \emph{deplSessions} in the provided pseudocode). When a session expires, the respective flow table entry is expected to expire and be deleted by the switch without any signalling between the controller, clients or servers\footnote{It is worth stressing that modelling such functionality is not supported by existing model checking approaches, such as \cite{mocs} and \cite{Kuai}, where flow table entries can only be explicitly deleted by the controller.}. As stated above, this controller program does not satisfy safety property \ref{eq:property} because the controller does nothing to rebalance the load when a session expires. Our model implementation\footnote{{\sc Uppaal} \cite{UPPAAL} is the back-end verification engine for MoCS and all experiments were run on an 18-Core iMac pro, 2.3GHz Intel Xeon W with 128GB DDR4 memory.} discovered the bug in the topology shown in Figure \ref{fig:lb} with 3 sessions in 11ms exploring 202 states.

In the second (still buggy) version of the controller, session scheduling is more sophisticated (shaded blue in \S\ref{subsection:cp}-CP\ref{alg:fwlb}); a session is assigned to the server with the least number of active sessions. Although the updated load balancing algorithm does keep track of the active sessions per server, this controller is still buggy because no rebalancing takes place when sessions expire. In a topology of 4 clients and 2 servers, we were able to discover the bug in 52ms after exploring 714 states.

We fix the bug by allowing the controller program to rebalance the active sessions, when 
	\begin{inparaenum}[(1)]
		\item a session expires and 
		\item the load is about to get out of balance, 
		\end{inparaenum}
		by moving one session from the most-loaded to the least-loaded server (\S\ref{subsection:cp}-CP\ref{alg:redistr}). In the same topology as above, we verified the property in 625ms after exploring 15068 states.\footnote{Note that the \emph{fsync}-optimisation was not enabled in the examples above.}

Next, we evaluate the performance of the proposed model and extensions for verifying the correctness of the property in a given SDN. We do that by verifying \ref{eq:property} with the correct controller program, discussed above, and scaling up the topology in terms of clients, servers and active sessions. Results are listed in Table~\ref{fig:tableexp} and state exploration is  illustrated in Figure~\ref{fig:lbexp}.

 Table~\ref{fig:tableexp} lists performance of the model checker for verifying the  correct controller program    with PORs disabled on the left  and with PORs enabled on the right, respectively. For each chosen topology we list the number of states explored, CPU  time used, and  memory used. The   topology is  shaped as in Figure~\ref{fig:lb}, and  parametrised by the number of clients (ranging from 3 to 5) and servers (ranging from 2 to 5), as indicated in Table \ref{fig:tableexp}. The number of required packets  and rules, respectively, is shown in grey. These numbers are always uniquely determined by the choice of topology.  Where there are no entries in the table (indicated by a dash) the verification did not terminate within 24 hours.
 
 The results clearly show that the verification scales well with the number of servers but not with the number of clients. 
 The reason for the latter is that for each additional client an additional packet is sent, which, according to programs \S\ref{subsection:cp}-CP\ref{alg:fwlb} and CP\ref{alg:redistr}, leads to 7 additional actions without timeouts and to 12 with timeouts. The causal ordering of these actions is shown in Fig.~\ref{fig:tree}. 
 The sub-branch in red shows the actions that appear due to a timeout of the added rule.
Thus, the number of states is exponential in the number of clients: every new  action in Fig.~\ref{fig:tree} leads to a new change of state, thus doubling the possible number of states.  This exponential blow-up happens  whether we have timeouts or not. With timeouts, however,  we  have worse exponential complexity as there are more new states generated.
 
 \begin{figure}[!htb]
 	\centering
 		\vspace{-.4cm}
{\small \begin{forest}
 	for tree={
 		edge+=->,
 	}, 
 	[$\mathit{send(pkt)}$
 		[$\mathit{nomatch(pkt)}$
 			[$\mathit{ctrl(pkt)}$ 	
 				 			[$\mathit{add(rule)}$
 			[$\mathit{match(pkt, rule)}$
 			]	 			
 			]
 				 			[$\mathit{add(rule_s)}$
 			[$\mathit{frmvd(rule_s)}$
 			[$\mathit{fsync(rule_s})$
 			[$\mathit{mod}(r)$, colour my roots=red
 			]	 			
 			[$\mathit{mod}(r_s)$, colour my roots=red
 			]	 			
 			]	 			
 			]	 			
 			]		
					 			[$\mathit{fwd(pkt)}$, colour my roots=black	 			
				[$\mathit{recv(pkt)}$
				]	 			
				]
 			] 		 		
 		]
 	]
 \end{forest}}
	 
	\caption{The causal enabling  relation between actions for an additional packet \emph{pkt}; only the  relevant arguments are shown using the same nomenclature as in the pseudocode.}
	 		\label{fig:tree}
\end{figure}
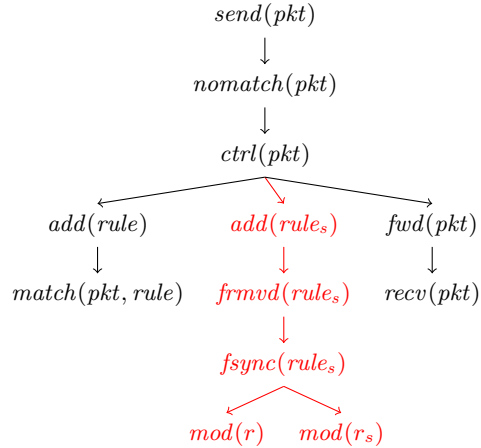

 The results also demonstrate that, for network setups with three clients, the POR optimisation reduces the state space  -- and thus the verification time --
 by about half. For more clients the reduction is far more significant, given that the verification of the unoptimised model did not terminate within 24 hours. This is not surprising as the number of possible interleavings is massively increased  by the non-deterministic timeout events.

\begin{figure}
	\centering
	\includegraphics[width=0.9\linewidth]{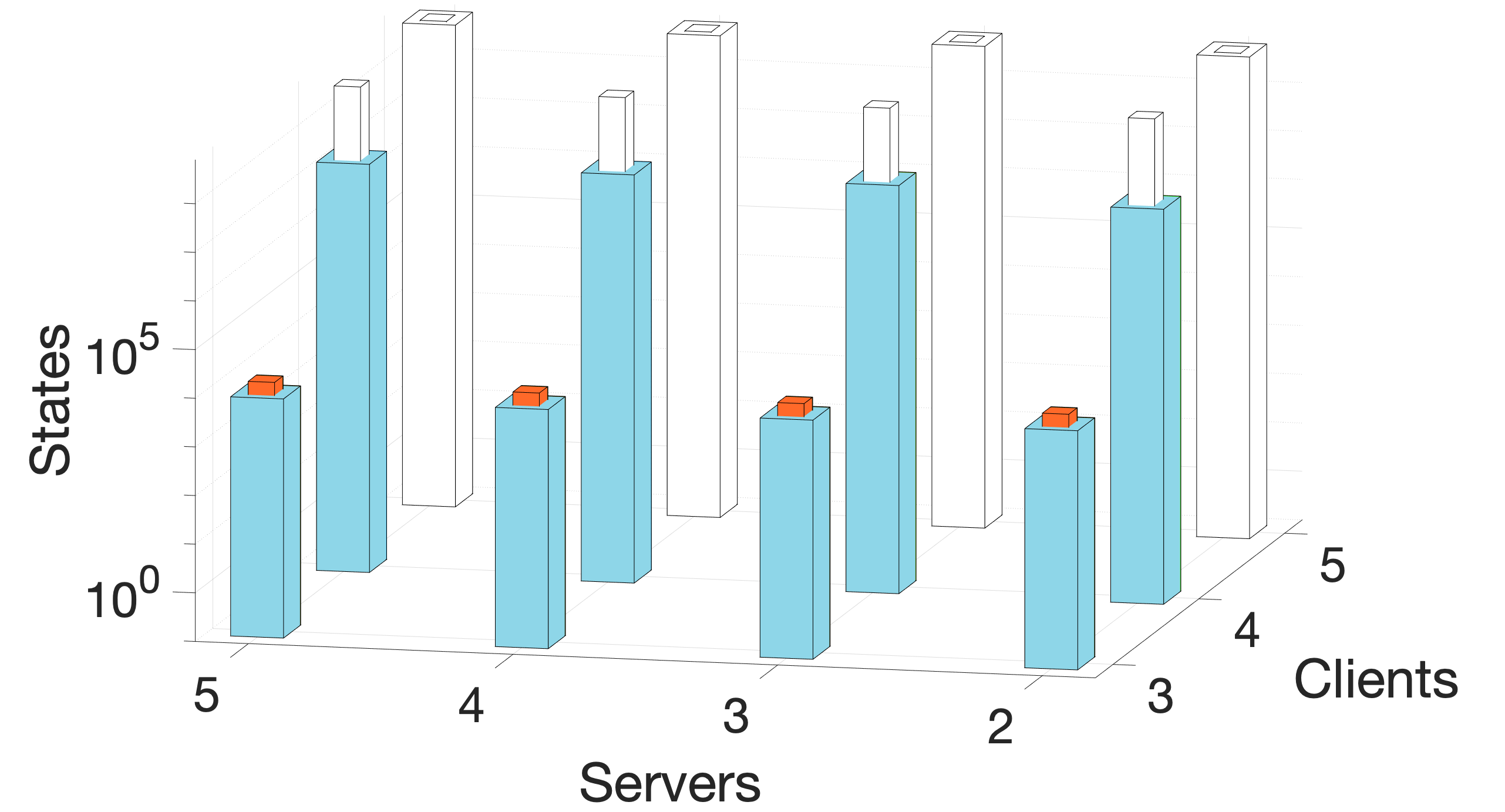}  
	\caption{Explored States (logarithmic scale). Wide bars represent the optimised model and narrow ones (inside) the unoptimised model.
		Uncoloured bars represent non-termination.}
	\label{fig:lbexp}
\end{figure}

\begin{table}
	\centering
		\caption{Performance by  number of clients and servers}
	\includegraphics[width=1\linewidth]{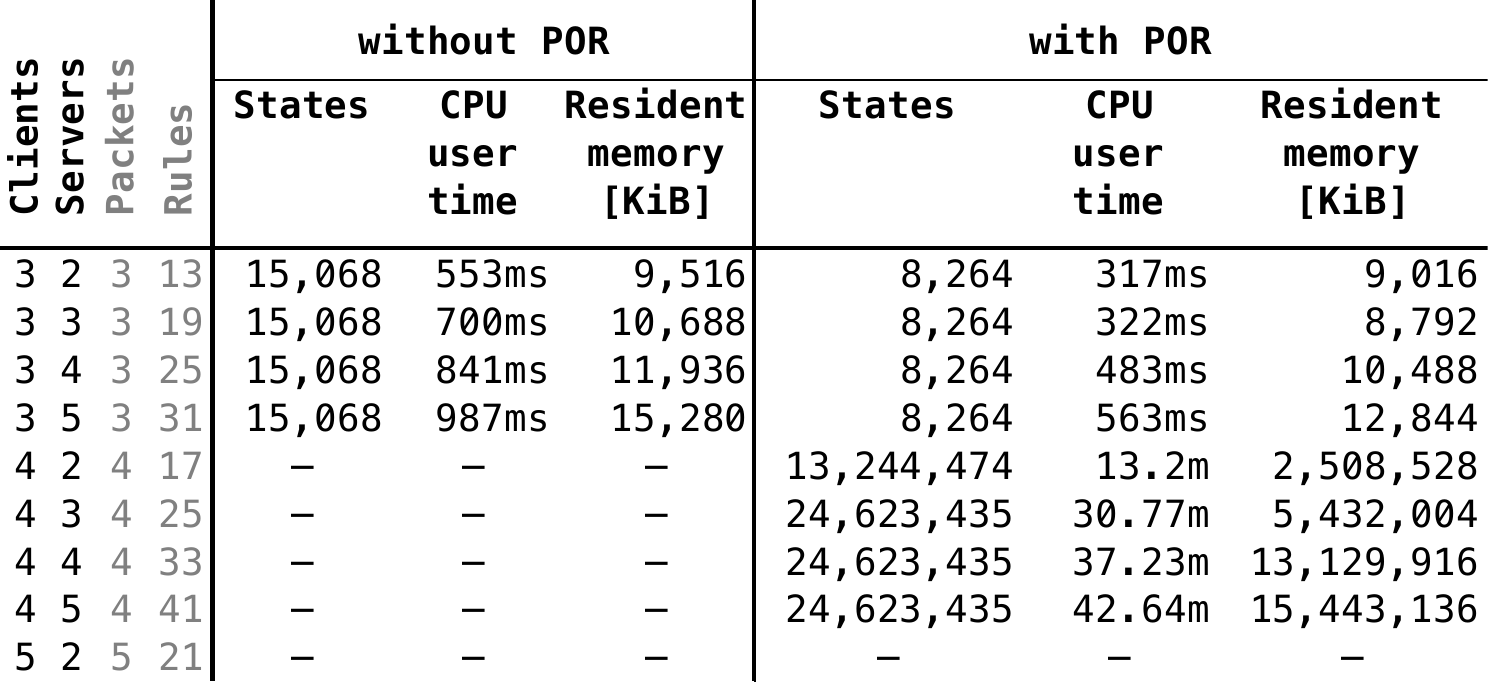}  
	\label{fig:tableexp}
	\vspace{-.3cm}
\end{table}

\section{Controller Programs}
\label{subsection:cp}

CP\ref{alg:fwlb} implements the \textit{PacketIn} message handler that processes packets sent by switches when the \textit{nomatch} action is fired. The two different versions of functionality discussed in the paper are defined by the \textit{leastConnectionsScheduling} constant. When \textit{leastConnectionsScheduling} is false, server selection is done in a round-robin fashion, whereas, in the opposite case, the controller assigns the new session to the server with the least number of active sessions.
%
%
\MyBox{starta}{enda}{blue}
\MyBox[-3.2ex]{startb}{endb}{gray}
\begin{algorithm}\small
	\caption{\emph{PacketIn} Message Handler}
	\algrenewcommand\algorithmicprocedure{\textbf{handler}}
	\begin{algorithmic}[1]
		\renewcommand{\algorithmicrequire}{\textbf{Input:}}
		\renewcommand{\algorithmicensure}{\textbf{Output:}}
		\Procedure{$\mathit{pktIn}$}{$\mathit{pkt,sw}$}
		
		
		\If{ $\mathit{pkt.srcIP} \neq \mathit{dodgy\_client} $ }
		
		\If{$\mathit{\neg deplSessions[pkt.srcIP]}$}
		
		\If{$\mathit{\neg leastConnectionsScheduling}$}
		
		\tikzmark{startb}
		~~~~~~~~\Comment{\textcolor{gray}{\small \texttt{Round-Robin rotation} }}\hfill \tikzmark{endb}
		
		\State $\mathit{server \leftarrow server} \bmod  2 + 1$ 
		
		\Else
		
		\tikzmark{starta}
		~~~~~~~~\Comment{\textcolor{gray}{\small \texttt{Least-Connections scheduling}}}\
		\State $\mathit{server \leftarrow min\big(sLoad[]\big)}$ \hfill \tikzmark{enda}
		
		\EndIf
		
		~~~~~~~~~\Comment{\textcolor{gray}{\small \texttt{Initialisation of flow to \emph{server}}}}
		
		\State $\mathrlap{\mathit{rule}.\mathit{srcIP}}\hphantom{\mathit{rule2.destTCPport}} \gets \mathit{pkt.srcIP}$
		\State $\mathrlap{\mathit{rule}.\mathit{in\_port}}\hphantom{\mathit{rule2.destTCPport}} \gets \mathit{pkt.in\_port}$
		\State $\mathrlap{\mathit{rule}.\mathit{fwdPort}}\hphantom{\mathit{rule2.destTCPport}} \gets \mathit{server}$
		
		~~~~~~~~~\Comment{\textcolor{gray}{\footnotesize \texttt{Initialisation of symmetric \emph{rule}\textsubscript{s}}}}
		
		\State $\mathrlap{\mathit{rule_s}.\mathit{srcIP}}\hphantom{\mathit{rule2.destTCPport}} \gets \mathit{server}$
		\State $\mathrlap{\mathit{rule_s}.\mathit{destIP}}\hphantom{\mathit{rule2.destTCPport}} \gets \mathit{pkt.srcIP}$
		
		
		\State $\mathrlap{\mathit{rule_s}.\mathit{fwdPort}}\hphantom{\mathit{rule2.destTCPport}} \gets \mathit{pkt.in\_port}$
		\State $\mathrlap{\mathit{rule_s}.\mathit{timeout}}\hphantom{\mathit{rule2.destTCPport}} \gets \mathit{true}$
		
		~~~~~~~~~\Comment{\textcolor{gray}{\footnotesize \texttt{Initialisation of drop rule \emph{rule}$_d$}}}
		
		\State $\mathrlap{\mathit{rule_d.srcIP}}\hphantom{\mathit{rule2.destTCPport}} \gets \mathit{dodgy\_client} $
		\State $\mathrlap{\mathit{rule_d.fwdPort}}\hphantom{\mathit{rule2.destTCPport}} \gets \mathit{drop}$

		~~~~~~~~~\Comment{\textcolor{gray}{\footnotesize \texttt{Deployment of rules}}}
		
		\State $\mathit{send\_message\big(FlowMod\big(add(rule)}\big), \mathit{sw\big)}$
		\State $\mathit{send\_message\big(FlowMod\big(add(rule_s)}\big), \mathit{sw\big)}$	
		\State $\mathit{send\_message\big(FlowMod\big(add(rule_d)\big), sw\big)}$	
		
		~~~~~~~~~\Comment{\textcolor{gray}{\footnotesize \texttt{Update firewall state table}}}
		
		\State $\mathit{sLoad[server]\texttt{++}}$
		\State $\mathit{deplSessions[pkt.srcIP]} \gets true$
		\EndIf
		
		~~~\Comment{\textcolor{gray}{\footnotesize \texttt{PacketOut: sending pkt out through sw}}}
		\State $\mathit{send\_message\{PacketOut(pkt, server), sw\big)}$
		\EndIf
		
		\EndProcedure
				
	\end{algorithmic}
	\label{alg:fwlb}
\end{algorithm}
%
%
%
%

CP\ref{alg:naive} implements the naive (and buggy) \textit{FlowRemoved} message handler. When soft state expires in the network, the handler merely updates its local state to reflect the update in the load.

\begin{algorithm}\small
	\caption{Naive \emph{FlowRemoved} message handler}
	\algrenewcommand\algorithmicprocedure{\textbf{handler}}
	\begin{algorithmic}[1]
		\renewcommand{\algorithmicrequire}{\textbf{Input:}}
		\renewcommand{\algorithmicensure}{\textbf{Output:}}
		
		\Procedure{$\mathit{flowRmvd}$}{$\mathit{rule_s,sw}$}
		
		\State $\mathit{sLoad[rule_s.srcIP]\texttt{--}}$  
		\State $\mathit{deplSessions[rule_s.destIP]} \gets \mathit{false}$
				
		\EndProcedure
		
	\end{algorithmic}
	\label{alg:naive}
\end{algorithm}
%
CP\ref{alg:redistr} implements a more sophisticated (and correct) \textit{FlowRemoved} message handler. When soft state expires in the network, the handler updates its local state to reflect the update in the load and re-assigns active sessions from the most to the least loaded server, by updating the flow table of the switch accordingly.

\begin{algorithm}\small
	\caption{Correct \emph{FlowRemoved} message handler}
	\algrenewcommand\algorithmicprocedure{\textbf{handler}}
	\begin{algorithmic}[1]
		\renewcommand{\algorithmicrequire}{\textbf{Input:}}
		\renewcommand{\algorithmicensure}{\textbf{Output:}}
		
		\Procedure{$\mathit{flowRmvd}$}{$\mathit{rule_s,sw}$}

		\State $\mathit{sLoad[rule_s.srcIP]\texttt{--}}$  
		\State $\mathit{deplSessions[rule_s.destIP]} \gets \mathit{false}$
		
		\If{$\mathit{max\big(sLoad[]\big) - min\big(sLoad[]\big)} > 1$}
		\State $\mathit{r \gets the~rule ~in~sw.ft~with~ fwdPort = max(sLoad[])}$
		\State $\mathit{r_s \gets symmetric ~rule ~of ~r}$
		\State $\mathit{cm \gets mod\big(r, fwdPort \gets min(sLoad[])\big)}$
		\State $\mathit{cm_s \gets mod\big(r_s, srcIP \gets min(sLoad[])\big)}$
		\State $\mathit{send\_message\big(FlowMod(cm, sw)\big)}$
		\State $\mathit{send\_message\big(FlowMod(cm_s, sw)\big)}$
		\State $\mathit{sLoad\big[max\big(sLoad[]\big)\big]\texttt{--}}$  
		\vspace{1mm}
		\State $\mathit{sLoad\big[min\big(sLoad[]\big)\big]\texttt{++}}$  
		
		\EndIf	
		\EndProcedure
		
	\end{algorithmic}
	\label{alg:redistr}
\end{algorithm}

\section{Conclusion and Future Work}
We have proposed model checking of SDN networks with flow entries (rules) that time out. Timeouts pose problems due to the great number of resulting interleavings to be explored. Our approach is the first one to deal with timeouts, exploiting partial-order reductions, and performing reasonably well for small networks. 
We demonstrated that bug finding works well for SDN networks in the presence of flow entry timeouts. Future work includes exploring  flow removals with timeouts that are constrained by integer to \emph{enforce certain orderings} of timeout messages as well as improvements in performance, for instance, by using bounded model checking tools for concurrent programs. 
%
%
%
%
%
\newpage



%

\bibliographystyle{IEEEtran}
\bibliography{ref}

\appendix

\noindent\emph{A ~~Proofs}
\label{appendix:proofs}

\begin{customlemma}{\ref{lemma:safe}}[{\sc Safeness}] For transition system $\mathcal{M}_{(\lambda,\textsc{cp})} = (S, A, \hookrightarrow, s_0, AP, L)$ and a formula $\varphi \in \text{LTL}_{ \setminus \{\bigcirc \}}$, 
	$\alpha = \mathit{fsync(sw, r, cs)}$  is safe   iff  the following two conditions are satisfied:
	
	\begin{quote}
	\begin{description}[leftmargin=!,labelwidth=\widthof{\bfseries Independence}]
		\item[Independence] {\sc cp} is not order-sensitive 
		\item[Invisibility] {\upshape\texttt{if} }$Q(q)$ in $\mathit{AP}$ occurs in $\varphi$, {\upshape\texttt{then}}  $\alpha$ is $\varphi$-invariant
	\end{description}	
	\end{quote}
		
\end{customlemma}

\begin{proof}
	To show safety we need to show two properties: \emph{independence} (action is independent of any other action) and \emph{invisibility} w.r.t.\ the context, in particular controller program, topology function and formula $\varphi$.\\
	
	\noindent\emph{Independence}:
	Recall that two actions $\alpha$ and $\beta \neq \alpha$ are independent iff for any state $s$ such that $\alpha\in A(s)$ and $\beta\in A(s)$:
	\begin{enumerate}[label=(\arabic*), leftmargin=.5in]
		\item $\alpha \in A(\beta(s))$ and $\beta\in A(\alpha(s))$  
		\item $\alpha(\beta(s))=\beta(\alpha(s))$
	\end{enumerate}
	\vspace{.6cm}
	\begin{enumerate}[label=(\arabic*)]
\item It can be easily checked that no instance of safe actions $\mathit{fsync}(\cdot)$ disables any other action, nor is any safe $\mathit{fsync}(\cdot)$ disabled by any other action, so the first condition of independence holds.
	\\
\item For any safe $\alpha = \mathit{fsync}(\cdot)$ and any other action $\beta$ we can assume already that they meet Condition (1). To show that any interleaving with any action $\beta \neq \alpha$ leads to the same state, we observe that

		\begin{itemize}[label=$\blacktriangleright$]
			\item if $\beta$ is not an \emph{fsync, ctrl} or \emph{bsync} action, then the mutations of queues by these actions do not interfere with each other. 
			\item The interesting cases occur when $\beta$ is in $\mathit{\{fsync(\cdot), ctrl(\cdot), bsync(\cdot)\}}$. 
			From the first condition we know that {\sc cp} is not order-sensitive, which implies that $\alpha$ and $\beta$ are independent.
			Order-insensitivity is a relatively strong condition but it ensures correctness of the lemma and thus partial order reduction.\footnote{Generalisations by a more clever analysis of the controller program are a future research topic.} Thus any interleaving of $\alpha$ and $\beta$ leads to the same state.

		\end{itemize}
	\end{enumerate}

	\vspace{.6cm}	
	
	\noindent\emph{Invisibility}: 
 $\alpha = \mathit{fsync(sw,r,cs)}$ may only affect \emph{frq}, $sw'.\mathit{fq}$, $sw'.cq$ (for some switches $sw'$), and the control state $\mathit{cs}$. 
		We know by definition of our Specification Language that an atomic proposition cannot refer to \emph{frq} or any $\mathit{fq}$, $\mathit{cq}$.
		In case the control state changes, $\alpha$ is invisible to $\varphi$ because of the second condition (\emph{Invisibility}) of Lemma \ref{lemma:safe}.
						
\end{proof}

\begin{customthm}{\ref{thm:frmvd}}[{\sc Flow-Removed Equivalence}]
	Given a property $\varphi \in \text{LTL}_{ \setminus \{\bigcirc \}}$, it holds that $\mathcal{M}_{(\lambda,\textsc{cp})}^{\mathit{fr}}$ satisfies $\varphi$ iff $\mathcal{M}_{(\lambda,\textsc{cp})}$ satisfies $\varphi$. 

\end{customthm}

\begin{proof} 
If $\mathit{ample}(s)$ satisfies the following conditions:
\begin{enumerate}[label=C\arabic*]
	\item (Non)emptiness condition: $\varnothing  \neq \mathit{ample(s)} \subseteq A(s)$.
	\item Dependency condition: Let $s \xhookrightarrow[]{\alpha_1} s_1...\xhookrightarrow[]{\alpha_n} s_n \xhookrightarrow[]{\beta} t$ be a run in $\mathcal{M}$. If $\beta \in A  \setminus \mathit{ample}(s)$ depends on $\mathit{ample}(s)$, then $\alpha_i \in \mathit{ample(s)}$ for some $0 < i \leq n$, which means that in every path fragment of $\mathcal{M} $, $\beta$ cannot appear before some transition from $\mathit{ample(s)}$ is executed.
	\item Invisibility condition: If $\mathit{ample}(s) \neq A(s)$ (i.e., state $s$ is not fully expanded), then every $\alpha \in \mathit{ample}(s)$ is invisible.
	\item Every cycle in $\mathcal{M}^{\mathit{fr}}$ contains a fully expanded state $s$ (i.e.\ $\mathit{ample}(s)=A(s)$).
\end{enumerate} 
then for each path in $\mathcal{M} $ there exists a stutter-trace equivalent path in $\mathcal{M}^{\mathit{fr}}$, and vice-versa, denoted $\mathcal{M}   \stutt \mathcal{M}^{\mathit{fr}}$ -- as we now show.

	$ $\newline
\vspace{-.4cm}
\begin{enumerate}[label=\textit{C\arabic*}]
	\item The (non)emptiness condition is trivial since by definition of $ample(s)$ it follows that $\mathit{ample}(s) = \varnothing $ iff $A(s) = \varnothing $.
	\item By assumption $\beta \in A  \setminus \mathit{ample}(s)$ depends on $\mathit{ample}(s)$. But with our definition of $\mathit{ample}(s)$ this is impossible as all actions in
	$\mathit{ample}(s)$ are safe and by definition independent of all other actions.
	
	%
	\item The validity of the invisibility condition is by definition of $\mathit{ample}$ and safe actions.
	\item 
	We now show that every cycle in $\mathcal{M}_{(\lambda,\textsc{cp})}^{\mathit{fr}}$ contains a fully expanded state $s$, i.e.\  a state $s$ such that $\mathit{ample}(s)=A(s)$.	
	\label{cycle}
	By definition of $\mathit{ample}(s)$ it is equivalent to show that there is no cycle in $\mathcal{M}_{(\lambda,\textsc{cp})}^{\mathit{fr}}$ consisting of safe actions only.
	We show this by contradiction, assuming such a cycle of only safe actions exists.

	Distinguish two cases.
	
	\begin{enumerate}[leftmargin=1.2cm,label=\textit{Case \arabic*}]
		\item A sequence of safe actions of same type.
 Let $\rho$ an execution of $\mathcal{M}_{(\lambda,\textsc{cp})}^{\mathit{fr}}$ which consists of only one type of \emph{fsync}-actions: 
		$\rho = s_1 \xhookrightarrow{\mathit{fsync}(\mathit{sw}_1, r_1,\mathit{cs}_1)}_{\mathit{fr}} s_2 \xhookrightarrow{\mathit{fsync}(\mathit{sw}_2,r_2,\mathit{cs}_2)}_{\mathit{fr}} ... s_{i-1} \xhookrightarrow{\mathit{fsync}(\mathit{sw}_{i-1}r_{i-1},,\mathit{cs}_{i-1})}_{\mathit{fr}}  s_i$.
		Suppose $\rho$ is a cycle. 
		According to the \emph{fsync} semantics, for each transition $s \xhookrightarrow{\mathit{fsync(sw,r,cs)}}_{\mathit{fr}} s' $, 
		where $\mathit{s = (\pi,\delta, \gamma)}$, $\mathit{s' = (\pi', \delta', \gamma')}$, it holds that $\mathit{\gamma'.frq = \gamma.frq\setminus \{r\}}$ as we use sets to represent \emph{frq} buffers. 
		Hence, for the execution $\rho$ it holds $\gamma_i.\mathit{frq} = \gamma_1.\mathit{frq}\setminus \{r_1, r_2,...r_{i-1}\}$ which implies that $s_1 \neq s_i$. Contradiction. 

\item A sequence of different safe actions. 
	Suppose there exists a cycle with mixed safe actions starting in $s_1$ and ending in $s_i$. 
	Distinguish the following cases.
	
	\begin{longenum}
		\item There exists at least a \emph{fsync} action in the cycle. According to the effects of safe transitions, the \emph{fsync} action will switch to a state with smaller \emph{frq}. It is important here that no action of other type than \emph{fsync} accesses \emph{frq}. This implies that $s_1 \neq s_i$. Contradiction.
		\item No \emph{fsync} action in the cycle.
		This is already established in \cite{mocs}. 										
	\end{longenum}
\end{enumerate}
\end{enumerate}

\end{proof}

\ignore{
CTL\textsuperscript{*}\cite{CTLast} subsumes both CTL and LTL.
Every LTL formula $\varphi$ is identified with the CTL\textsuperscript{*} formula $\forall \varphi$, so we take as reference the semantics of LTL over states.

Each LTL formula $\varphi$ can be embedded into CTL\textsuperscript{*}: $\mathcal{M} \models \varphi \Leftrightarrow \mathcal{M} \models \forall \varphi$ \cite{Baier2008}

}
\end{document}